\documentclass[11pt]{article}
\usepackage{hyperref}
\usepackage{times}  
\usepackage{mathpazo}
\usepackage{amssymb,amsmath,amsthm}
\usepackage{epsfig}
\usepackage{tcolorbox}
\usepackage{lipsum} 
\usepackage{enumitem}

\newtheorem{theorem}{Theorem}[section]
\newtheorem{proposition}[theorem]{Proposition}
\newtheorem{definition}[theorem]{Definition}

\newtheorem{lemma}[theorem]{Lemma}

\newcommand{\qedsymb}{\hfill{\rule{2mm}{2mm}}}
\renewenvironment{proof}[1][]{\begin{trivlist}
\item[\hspace{\labelsep}{\bf\noindent Proof#1:\/}] }{\qedsymb\end{trivlist}}

\def\calI{{\cal I}}
\def\calC{{\cal C}}

\def\N{\mathbb{N}}

\newcommand{\Crown}{\mathsf{Crown}}

\newcommand{\NP}{\mathsf{NP}}

\renewcommand{\epsilon}{\varepsilon}

\newcommand{\rank}{\mathop{\mathrm{rank}}}
\newcommand{\minrank}{\mathop{\mathrm{minrk}}}
\newcommand{\Capa}{\mathop{\mathrm{Cap}}}
\newcommand{\Ind}{\mathop{\mathrm{Ind}}}

\newcommand{\Fset}{\mathbb{F}}         

\newcommand{\Conf}{{\mathfrak{C}}}

\newcommand{\IC}{\textsc{Index Coding}}
\newcommand{\DIC}{\textsc{Dual Index Coding}}
\newcommand{\SC}{\textsc{Storage Capacity}}

\newcommand{\DMR}{\textsc{Dual \linebreak[0] Minrank}}

\begin{document}

\title{{\bf Kernels for Storage Capacity and Dual Index Coding}}

\author{
Ishay Haviv\thanks{School of Computer Science, The Academic College of Tel Aviv-Yaffo, Tel Aviv 61083, Israel. Research supported in part by the Israel Science Foundation (grant No.~1218/20).}
}

\date{}

\maketitle

\begin{abstract}
The storage capacity of a graph measures the maximum amount of information that can be stored across its vertices, such that the information at any vertex can be recovered from the information stored at its neighborhood. The study of this graph quantity is motivated by applications in distributed storage and by its intimate relations to the index coding problem from the area of network information theory. In the latter, one wishes to minimize the amount of information that has to be transmitted to a collection of receivers, in a way that enables each of them to discover its required data using some prior side information.

In this paper, we initiate the study of the $\SC$ and $\IC$ problems from the perspective of parameterized complexity. We prove that the $\SC$ problem parameterized by the solution size admits a kernelization algorithm producing kernels of linear size. We also provide such a result for the $\IC$ problem, in the linear and non-linear settings, where it is parameterized by the dual value of the solution, i.e., the length of the transmission that can be saved using the side information. A key ingredient in the proofs is the crown decomposition technique due to Chor, Fellows, and Juedes (WG~2003, WG~2004). As an application, we significantly extend an algorithmic result of Dau, Skachek, and Chee (IEEE Trans. Inform. Theory,~2014).
\end{abstract}

\section{Introduction}

The storage capacity of a graph, introduced in 2014 by Mazumdar~\cite{Mazumdar15} and by Shanmugam and Dimakis~\cite{ShanmugamD14}, measures the maximum amount of information that can be stored across its vertices, such that the information at any vertex can be recovered from the information stored at its neighborhood.
This graph quantity generalizes the well-studied notion of locally repairable codes and is motivated by applications in distributed storage, in which some data is stored across multiple servers and one desires to be able to recover the information of any server from the information of accessible servers. The topology of the distributed system is given by a graph whose vertices represent the servers, where two servers are adjacent if they are accessible to each other due to their physical locations or the architecture of the system.
Formally speaking, the storage capacity $\Capa_q(G)$ of a graph $G$ on the vertex set $[n] = \{1,2,\ldots,n\}$ over an alphabet of size $q \geq 2$ is defined as the maximum of $\log_q |C|$, taken over all codes $C \subseteq [q]^n$ such that for every codeword $x \in C$ and for every vertex $i \in [n]$, the value of $x_i$ can be determined from the restriction $x_{N_G(i)}$ of $x$ to the entries of the neighborhood $N_G(i)$ of $i$.

The storage capacity of graphs turns out to be closely related to the problem of index coding with side information, a central problem in network information theory that was introduced by Birk and Kol~\cite{BirkKol98} in 1998.
This problem involves $n$ receivers $R_1, \ldots, R_n$, such that the $i$th receiver $R_i$ is interested in the $i$th symbol $x_i$ of an $n$-symbol message $x \in \Sigma^n$ over an alphabet $\Sigma$ and has prior side information that consists of some symbols $x_j$ with $j \in [n] \setminus \{i\}$. The side information map, which is assumed here to be symmetric, can be represented by a graph $G$ on the vertex set $[n]$, where two vertices $i$ and $j$ are adjacent if $R_i$ knows $x_j$ and $R_j$ knows $x_i$.
An index code of length $\ell$ over $\Sigma$ for $G$ is an encoding function $E: \Sigma^n \rightarrow \Sigma^\ell$, such that for every $x \in \Sigma^n$, broadcasting the codeword $E(x)$ will enable each receiver $R_i$ to discover $x_i$ based on $E(x)$ and on its side information $x_{N_G(i)}$. For a side information graph $G$ and an alphabet $\Sigma$ of size $q$, let $\Ind_q(G)$ denote the minimum possible length $\ell$ of an index code over $\Sigma$ for $G$. An index code is called linear if its alphabet is a field $\Fset$ and the encoding function is linear over $\Fset$. It was shown by Bar-Yossef, Birk, Jayram, and Kol~\cite{BBJK06} that the minimum possible length of a linear index code over a field $\Fset$ for a graph $G$ coincides with a graph quantity called minrank and denoted by ${\minrank}_\Fset(G)$, which was introduced in the late seventies by Haemers~\cite{Haemers78} in the study of the Shannon capacity of graphs (see Definition~\ref{def:minrank}).

The storage capacity and index coding problems are related in a dual manner.
It was shown in~\cite{Mazumdar15} (see also~\cite{AlonLSWH08}) that for every graph $G$ on $n$ vertices and for every fixed integer $q \geq 2$, the storage capacity ${\Capa}_q(G)$ differs from the quantity $n-{\Ind}_q(G)$ by no more than an additive term logarithmic in $n$, and that there exist graphs for which they are not equal. The quantity $n-{\Ind}_q(G)$ can be viewed as a dual (or complementary) value of the minimum length of an index code for $G$ over an alphabet of size $q$, expressing the maximum number of symbols that can be saved by an index code, in comparison to transmitting the entire message.
It was further shown in~\cite{ShanmugamD14} that the linear variant of the storage capacity of $G$, where the codewords are required to form a linear subspace over a field $\Fset$, precisely coincides with the dual value of the optimal length of a linear index code over $\Fset$ for $G$, that is, $n-{\minrank}_\Fset(G)$.

The computational tasks of determining and approximating the values of $\Capa_q(G)$, $\Ind_q(G)$, and ${\minrank}_\Fset(G)$ for an input graph $G$ were extensively studied in the literature.
Mazumdar~\cite{Mazumdar15} proved that there exists a polynomial-time algorithm that approximates the value of ${\Capa}_q(G)$ for a given graph $G$ to within a factor of $2$.
It was further shown by Mazumdar, McGregor, and Vorotnikova~\cite{MazumdarMV19} that given a graph $G$, it is possible to exactly determine the values of ${\Capa}_q(G)$ and $\Ind_q(G)$ in double-exponential time, provided that $q \geq 2$ is a fixed constant (see Proposition~\ref{prop:exp_algo}).
On the hardness side, Langberg and Sprintson~\cite{LangbergS08} proved that for every constant $q \geq 2$, it is impossible to efficiently approximate the value of $\Ind_q(G)$ for an input graph $G$ to within any constant factor, assuming a certain variant of the Unique Games Conjecture. As for the minrank over a given field $\Fset$, the problem of deciding whether an input graph $G$ satisfies ${\minrank}_\Fset(G) \leq k$ is solvable in polynomial time for $k \in \{1,2\}$, but it is $\NP$-hard for any $k \geq 3$, as was proved by Peeters~\cite{Peeters96} already in 1996.
More recently, it was proved in~\cite{ChawinH23} that for every finite field $\Fset$, it is $\NP$-hard to approximate ${\minrank}_\Fset(G)$ for a given graph $G$ to within any constant factor (see also~\cite{ChawinH24}).
Yet, for certain graph families of interest, it is possible to determine the above quantities in polynomial time or to approximate them to within tighter factors (see~\cite{BerlinerL11,Arbabjolfaei016,AgarwalFM19,ChawinH24a}).
The complexity of the index coding problem on instances with near-extreme values was studied by Dau, Skachek, and Chee~\cite{DauSC14}.
Among other things, they showed that for $k \in \{0,1,2\}$, there exist polynomial-time algorithms that decide whether a given graph $G$ on $n$ vertices satisfies ${\Ind}_q(G)=n-k$ and whether it satisfies ${\minrank}_\Fset(G) = n-k$.

Before turning to our contribution, let us briefly present a couple of fundamental concepts from the area of parameterized complexity (see, e.g.,~\cite{CyganFKLMPPS15,KernelBook19}). This area investigates the complexity of problems with respect to specified parameters of their inputs.
A parameterized problem is a collection $L$ of pairs $(x,k) \in \Sigma^* \times \N$ for a fixed finite alphabet $\Sigma$, where $k$ is referred to as the parameter of the problem.
A fixed-parameter algorithm for $L$ is an algorithm that given an instance $(x,k)$ decides whether $(x,k) \in L$ in time $|x|^c \cdot f(k)$, where $c$ is an absolute constant and $f: \N \rightarrow \N$ is a computable function.
A kernelization algorithm for $L$ is an algorithm that given an instance $(x,k)$ runs in polynomial time and returns an instance $(x',k')$, called a kernel, satisfying $|x'|,k' \leq h(k)$ for some computable function $h: \N \rightarrow \N$, such that the instances $(x,k)$ and $(x',k')$ are equivalent, in the sense that $(x,k) \in L$ if and only if $(x',k') \in L$. Note that a fixed-parameter algorithm for $L$ can be obtained by first applying a kernelization algorithm to reduce a given instance to an equivalent one whose size is bounded by a function of the parameter, and then calling some algorithm that decides $L$ on the reduced instance.

\subsection{Our Contribution}

The present paper initiates the study of the storage capacity and index coding problems from the perspective of parameterized complexity.
For an integer $q \geq 2$, consider the $\SC_q$ problem parameterized by the solution size, namely, the parameterized problem that given a graph $G$ and an integer $k \geq 0$ asks to decide whether ${\Capa}_q(G) \geq k$, where $k$ is the parameter. We prove the following kernelization result.

\begin{theorem}\label{thm:IntroKernelCap}
For every integer $q \geq 2$, there exists a (polynomial-time) kernelization algorithm for the $\SC_q$ problem, which given an instance $(G,k)$ returns an equivalent instance $(G',k')$ where $G'$ has at most $\max(3k'-3,0)$ vertices and $k' \leq k$.
\end{theorem}

We next consider the computational index coding problem parameterized by the dual value of the solution size.
Formally, for an integer $q \geq 2$, let $\DIC_q$ denote the parameterized problem that given a graph $G$ on $n$ vertices and an integer $k \geq 0$ asks to decide whether ${\Ind}_q(G) \leq n-k$, where $k$ is the parameter. We prove the following kernelization result.

\begin{theorem}\label{thm:IntroKernelInd}
For every integer $q \geq 2$, there exists a (polynomial-time) kernelization algorithm for the $\DIC_q$ problem, which given an instance $(G,k)$ returns an equivalent instance $(G',k')$ where $G'$ has at most $\max(3k'-3,0)$ vertices and $k' \leq k$.
\end{theorem}

We finally consider the computational problem associated with the minrank quantity over any given field, parameterized by its dual value.
For a field $\Fset$, let $\DMR_\Fset$ denote the parameterized problem that given a graph $G$ on $n$ vertices and an integer $k \geq 0$ asks to decide whether ${\minrank}_\Fset(G) \leq n-k$, where $k$ is the parameter. We prove the following kernelization result.

\begin{theorem}\label{thm:IntroKernelMR}
For every field $\Fset$, there exists a (polynomial-time) kernelization algorithm for the \linebreak $\DMR_\Fset$ problem, which given an instance $(G,k)$ returns an equivalent instance $(G',k')$ where $G'$ has at most $\max(3k'-3,0)$ vertices and $k' \leq k$.
\end{theorem}

A key ingredient in the proofs of Theorems~\ref{thm:IntroKernelCap},~\ref{thm:IntroKernelInd}, and~\ref{thm:IntroKernelMR} is the crown decomposition technique due to Chor, Fellows, and Juedes~\cite{Fellows03,ChorFJ04} (see Section~\ref{sec:crown}). Their approach was applied over the years in kernelization algorithms for various problems, e.g., vertex cover, dual coloring, maximum satisfiability, and longest cycle (see, e.g.,~\cite[Section~3]{JacobMR23}). To obtain our results, we essentially borrow the algorithms of~\cite{ChorFJ04} for the vertex cover and dual coloring problems, and analyze them with respect to the graph quantities studied in the present paper.
For a comprehensive introduction to the concept of crown decomposition, the reader is referred to~\cite[Chapter~4]{KernelBook19}.

The graphs produced by our kernelization algorithms are of linear size in the parameter, namely, they have at most $\max(3k-3,0)$ vertices for instances with parameter $k$.
Combining this fact with the double-exponential time algorithms of~\cite{MazumdarMV19} for the storage capacity and index coding problems and with a brute-force exponential time algorithm for minrank, we obtain the following fixed-parameter tractability results. Here and throughout, the $\widetilde{O}$ notation is used to hide poly-logarithmic multiplicative factors.

\begin{theorem}\label{thm:IntroFPT}
Let $q \geq 2$ be an integer, and let $\Fset$ be a finite field.
\begin{enumerate}
  \item\label{itm:FPT1} There exists an algorithm for the $\SC_q$ problem whose running time on a graph on $n$ vertices and an integer $k$ is $n^{O(1)} \cdot \widetilde{O}(1.1996^{q^{3k}})$.
  \item\label{itm:FPT2} There exists an algorithm for the $\DIC_q$ problem whose running time on a graph on $n$ vertices and an integer $k$ is $n^{O(1)} \cdot \widetilde{O}(2^{q^{3k}})$.
  \item\label{itm:FPT3} There exists an algorithm for the $\DMR_\Fset$ problem whose running time on a graph on $n$ vertices and an integer $k$ is $n^{O(1)} \cdot |\Fset|^{9k^2}$.
\end{enumerate}
\end{theorem}

Note that Items~\ref{itm:FPT2} and~\ref{itm:FPT3} of Theorem~\ref{thm:IntroFPT} imply that for every constant integer $k$, it is possible to decide in polynomial time whether a given graph $G$ on $n$ vertices satisfies ${\Ind}_q(G) = n-k$ and whether it satisfies ${\minrank}_\Fset(G) = n-k$, provided that $q \geq 2$ is a fixed integer and that $\Fset$ is a fixed finite field.
In fact, this holds even when $k$ is not a constant but grows sufficiently slowly with $n$, namely, for certain $k = O(\log \log n)$ in the former problem and for any $k = O(\sqrt{\log n})$ in the latter. This significantly extends the aforementioned result of Dau et al.~\cite{DauSC14}, who provided a polynomial-time algorithm for these problems where $k \in \{0,1,2\}$.

To conclude, let us stress that the present paper focuses on the storage capacity and index coding problems for undirected graphs, whereas both problems have natural extensions of interest to directed graphs (see, e.g.,~\cite{Mazumdar15}). It would thus be interesting to decide whether our results can be generalized to the directed setting.
An additional avenue for future research would be to examine the parameterized complexity of these problems with respect to other parameterizations (see, e.g.,~\cite{HR24}).

\subsection{Outline}
The rest of the paper is organized as follows.
In Section~\ref{sec:preliminaries}, we collect several notations, definitions, and tools that will be used throughout the paper.
In Section~\ref{sec:kernels}, we present and analyze the kernelization algorithm for the $\SC_q$, $\DIC_q$, and $\DMR_\Fset$ problems and prove Theorems~\ref{thm:IntroKernelCap},~\ref{thm:IntroKernelInd}, and~\ref{thm:IntroKernelMR}. We also derive there the fixed-parameter tractability of these problems and confirm Theorem~\ref{thm:IntroFPT}.

\section{Preliminaries}\label{sec:preliminaries}

\subsection{Notations}
Throughout this paper, all graphs are undirected and simple.
We let $K_0$ denote the graph with no vertices and no edges.
For a graph $G=(V,E)$ and a vertex $u \in V$, we let $N_G(u)$ denote the set of vertices that are adjacent to $u$ in $G$. The vertex $u$ is said to be isolated in $G$ if it satisfies $N_G(u) = \emptyset$. For a set $S \subseteq V$, we let $G[S]$ stand for the subgraph of $G$ induced by $S$. A matching of $G$ is a set $M \subseteq E$ of pairwise disjoint edges. As usual, we let $\alpha(G)$ and $\chi(G)$ denote the independence and chromatic numbers of $G$ respectively. We also let $\overline{\chi}(G) = \chi(\overline{G})$ denote the clique cover number of $G$.
For an integer $n$, let $[n] = \{1,2,\ldots, n\}$. For a vector $x \in \Sigma^n$ over some alphabet $\Sigma$ and for a set $S \subseteq [n]$, let $x_S$ denote the restriction of $x$ to the entries whose indices lie in $S$.

\subsection{Storage Capacity, Index Coding, and Minrank}

The optimal solutions of the storage capacity and index coding problems can both be characterized using the notion of confusion graphs, introduced in~\cite{BBJK06} and further studied in~\cite{AlonLSWH08,Mazumdar15,MazumdarMV19}, defined as follows.

\begin{definition}[Confusion Graph]
For a graph $G = ([n],E)$ and an integer $q \geq 2$, the {\em confusion graph} of $G$ over $[q]$, denoted by $\Conf_q(G)$, is the graph on the vertex set $[q]^n$, where two distinct vertices $x,y$ are adjacent if for some $i \in [n]$, it holds that $x_i \neq y_i$ and yet $x_{N_G(i)} = y_{N_G(i)}$.
\end{definition}

The storage capacity of a graph $G=([n],E)$ over the alphabet $[q]$ for some $q \geq 2$ is denoted by $\Capa_q(G)$ and is defined as the maximum of $\log_q |C|$, taken over all codes $C \subseteq [q]^n$ such that for every codeword $x \in C$ and for every vertex $i \in [n]$, the value of $x_i$ can be determined from $x_{N_G(i)}$. Observe that two distinct vectors of $[q]^n$ can be members of such a code simultaneously if and only if they are not adjacent in the confusion graph $\Conf_q(G)$. This yields the following characterization of $\Capa_q(G)$.

\begin{proposition}[\cite{Mazumdar15}]\label{prop:CapConf}
For every graph $G$ and for every integer $q \geq 2$, $\Capa_q(G) = \log_q \alpha (\Conf_q(G))$.
\end{proposition}

An index code of length $\ell$ over an alphabet $\Sigma$ for a graph $G=([n],E)$ is an encoding function $E:\Sigma^n \rightarrow \Sigma^\ell$ such that for each $i \in [n]$ there exists a decoding function $g_i:\Sigma^{\ell+|N_G(i)|} \rightarrow \Sigma$ satisfying $g_i(E(x),x_{N_G(i)})=x_i$ for all $x \in \Sigma^n$. For an integer $q \geq 2$, let $\Ind_q(G)$ denote the smallest integer $\ell$ for which there exists an index code of length $\ell$ over $[q]$ for $G$.
Observe that an index code over $[q]$ for $G$ can encode two distinct vectors of $[q]^n$ by the same codeword if and only if they are not adjacent in the confusion graph $\Conf_q(G)$. Therefore, every such index code corresponds to a proper coloring of $\Conf_q(G)$, where each color class is associated with a distinct codeword. This yields the following characterization of $\Ind_q(G)$ (see, e.g.,~\cite{BBJK06,AlonLSWH08}).

\begin{proposition}\label{prop:IndConf}
For every graph $G$ and for every integer $q \geq 2$, $\Ind_q(G) = \lceil \log_q \chi (\Conf_q(G)) \rceil$.
\end{proposition}

The following lemma gives a well-known upper bound on $\Ind_q(G)$.
\begin{lemma}[\cite{BirkKol98}]\label{lemma:IndCliqueCover}
For every graph $G$ and for every integer $q \geq 2$, $\Ind_q(G) \leq \overline{\chi}(G)$.
\end{lemma}
\begin{proof}
For a graph $G=([n],E)$, let $\calC$ be a clique cover of $G$ of minimum size.
Consider the index code that encodes any vector $x \in [q]^n$ by the vector of length $|\calC|$ that for each clique $C \in \calC$ includes the sum $\sum_{i \in C}{x_i}$ modulo $q$ (represented by an element of $[q]$). Observe that for each $i \in [n]$, the value of $x_i$ can be recovered given $x_{N_G(i)}$ and given the value associated with the clique of $\calC$ that includes the vertex $i$. This implies that $\Ind_q(G) \leq |\calC| = \overline{\chi}(G)$, as required.
\end{proof}

The following lemma presents one direction of the duality, shown in~\cite{Mazumdar15}, between storage capacity and index coding.

\begin{lemma}[\cite{Mazumdar15}]\label{lemma:DualityIndCap}
For every graph $G$ on $n$ vertices and for every integer $q \geq 2$, it holds that
\[{\Ind}_q(G) \geq n-{\Capa}_q(G).\]
\end{lemma}
\begin{proof}
Let $G$ be a graph on $n$ vertices, and fix some integer $q \geq 2$.
By combining Propositions~\ref{prop:CapConf} and~\ref{prop:IndConf} with the fact that every graph $H=(V_H,E_H)$ satisfies $\alpha(H) \cdot \chi(H) \geq |V_H|$, we obtain that
\[ {\Ind}_q(G) \geq \log_q \chi(\Conf_q(G)) \geq \log_q \Big (\frac{q^{n}}{\alpha(\Conf_q(G))} \Big )= n-{\Capa}_q(G),\]
so we are done.
\end{proof}

We next define the minrank parameter of graphs, introduced by Haemers~\cite{Haemers78}.

\begin{definition}[Minrank]\label{def:minrank}
Let $G=([n],E)$ be a graph, and let $\Fset$ be a field.
We say that a matrix $A \in \Fset^{n \times n}$ {\em represents} $G$ over $\Fset$ if $A_{i,i} \neq 0$ for every $i \in [n]$, and $A_{i,j}=0$ for every pair $i,j \in [n]$ of distinct vertices that are not adjacent in $G$.
The {\em minrank} of $G$ over $\Fset$ is defined by
\[{\minrank}_\Fset(G) =  \min\{{\rank}_{\Fset}(A)\mid A \mbox{ represents }G\mbox{ over }\Fset\}.\]
\end{definition}

It is easy to see that the clique cover number forms an upper bound on minrank.

\begin{lemma}[\cite{Haemers78}]\label{lemma:MRCliqueCover}
For every graph $G$ and for every field $\Fset$, ${\minrank}_\Fset(G) \leq \overline{\chi}(G)$.
\end{lemma}
\begin{proof}
For a graph $G=([n],E)$, let $\calC$ be a clique cover of $G$ of minimum size.
Consider the matrix $A \in \Fset^{n \times n}$ defined by $A_{i,j} = 1$ if $i$ and $j$ are (possibly equal) members of the same clique of $\calC$, and $A_{i,j}=0$ otherwise.
Observe that $A$ represents the graph $G$ over $\Fset$ and that ${\rank}_\Fset(A)=|\calC|$. This implies that ${\minrank}_\Fset(G) \leq |\calC| = \overline{\chi}(G)$, as required.
\end{proof}

We will need the following simple lemma.

\begin{lemma}\label{lemma:union}
Let $G = (V,E)$ be a graph, and let $V= V_1 \cup V_2$ be a partition of its vertex set. Then, for every integer $q \geq 2$,
${\Capa}_q(G) \geq {\Capa}_q(G[V_1])+{\Capa}_q(G[V_2])$ and ${\Ind}_q(G) \leq {\Ind}_q(G[V_1])+{\Ind}_q(G[V_2])$, and for every field $\Fset$, ${\minrank}_\Fset(G) \leq {\minrank}_\Fset(G[V_1])+{\minrank}_\Fset(G[V_2])$.
\end{lemma}

\begin{proof}
Observe that if $I_1$ and $I_2$ are independent sets in the confusion graphs $\Conf_q(G[V_1])$ and $\Conf_q(G[V_2])$ respectively, then the set $I_1 \circ I_2$ of all the vectors obtained by concatenating a vector of $I_1$ to a vector of $I_2$ forms an independent set in the confusion graph $\Conf_q(G)$. This implies that $\alpha(\Conf_q(G)) \geq \alpha(\Conf_q(G[V_1]))  \cdot \alpha(\Conf_q(G[V_2]))$, which using Proposition~\ref{prop:CapConf}, yields that
\[{\Capa}_q(G) \geq {\Capa}_q(G[V_1])+{\Capa}_q(G[V_2]).\]
Observe further that if $\calI_1$ and $\calI_2$ are collections of color classes of proper colorings of $\Conf_q(G[V_1])$ and $\Conf_q(G[V_2])$ respectively, then the collection $\{I_1 \circ I_2 \mid I_1 \in \calI_1,~I_2 \in \calI_2\}$ forms a partition of the vertex set of $\Conf_q(G)$ into independent sets, hence $\chi(\Conf_q(G)) \leq \chi(\Conf_q(G[V_1]))  \cdot \chi(\Conf_q(G[V_2]))$. Using Proposition~\ref{prop:IndConf}, it follows that
\begin{eqnarray*}
{\Ind}_q(G) &=& \lceil \log_q \chi(\Conf_q(G)) \rceil \leq \lceil \log_q \chi(\Conf_q(G[V_1])) + \log_q \chi(\Conf_q(G[V_2])) \rceil \\
&\leq& \lceil \log_q \chi(\Conf_q(G[V_1]))\rceil + \lceil \log_q \chi(\Conf_q(G[V_2])) \rceil = {\Ind}_q(G[V_1])+{\Ind}_q(G[V_2]).
\end{eqnarray*}
Finally, observe that if $A_1$ and $A_2$ are matrices that represent over a field $\Fset$ the graphs $G[V_1]$ and $G[V_2]$ respectively, then the $2 \times 2$ block matrix $A$ that has the matrices $A_1$ and $A_2$ as diagonal blocks and zeros elsewhere, represents $G$ over $\Fset$ and satisfies ${\rank}_\Fset(A) = {\rank}_\Fset(A_1)+{\rank}_\Fset(A_2)$. This implies that ${\minrank}_\Fset(G) \leq {\minrank}_\Fset(G[V_1])+{\minrank}_\Fset(G[V_2])$, and we are done.
\end{proof}

The following proposition provides algorithms for exactly determining the $\Capa_q(G)$, $\Ind_q(G)$, and ${\minrank}_\Fset(G)$ quantities for an input graph $G$, provided that $q \geq 2$ and that $\Fset$ is a finite field.

\begin{proposition}\label{prop:exp_algo}
Let $q \geq 2$ be an integer, and let $\Fset$ be a finite field.
\begin{enumerate}
  \item\label{alg:C} There exists an algorithm that given a graph $G$ on $n$ vertices determines $\Capa_q(G)$ in time $\widetilde{O}(1.1996^{q^n})$.
  \item\label{alg:I} There exists an algorithm that given a graph $G$ on $n$ vertices determines $\Ind_q(G)$ in time $\widetilde{O}(2^{q^n})$.
  \item\label{alg:MR} There exists an algorithm that given a graph $G$ on $n$ vertices determines ${\minrank}_\Fset(G)$ in time $\widetilde{O} (|\Fset|^{n^2})$.
\end{enumerate}
\end{proposition}
\noindent
The first two items of Proposition~\ref{prop:exp_algo}, given in~\cite[Theorem~2]{MazumdarMV19}, follow by constructing the confusion graph $\Conf_q(G)$ and calling known algorithms for computing the independence and chromatic numbers~\cite{XiaoN17,BjorklundHK09}. The third item follows by simply enumerating all the $n \times n$ matrices over $\Fset$ and finding a matrix of minimum rank that represents $G$.

\subsection{Crown Decomposition}\label{sec:crown}

Our kernelization algorithms use the notion of crown decomposition, introduced in~\cite{Fellows03,ChorFJ04} and defined as follows (see also~\cite[Chapter~4]{KernelBook19}). For an illustration, see Figure~\ref{fig:crown}.
\begin{definition}[Crown Decomposition]\label{def:crown}
A {\em crown decomposition} of a graph $G=(V,E)$ is a partition of its vertex set into three sets $V = C \cup H \cup R$ where $C,H \neq \emptyset$, such that
\begin{enumerate}
  \item $C$ is an independent set,
  \item $H$ separates $C$ and $R$, that is, no edge of $G$ has one endpoint in $C$ and one endpoint in $R$, and
  \item there exists a matching of $H$ into $C$, that is, the subgraph of $G$ spanned by the edges that have one endpoint in $H$ and one endpoint in $C$ admits a matching of size $|H|$.
\end{enumerate}
\end{definition}

\begin{figure}[h]
\center{\epsfxsize=2.5in\epsfbox{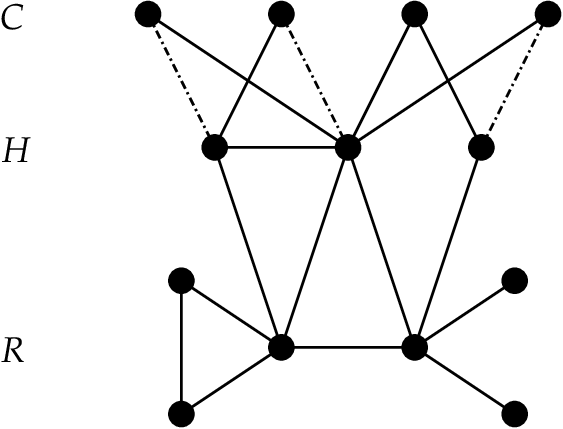}}
 \caption{A crown decomposition. The vertices of $C$ (Crown) form an independent set, and they are separated by the vertices of $H$ (Head) from those of $R$ (Royal body). The marked edges form a matching of $H$ into $C$.}
 \label{fig:crown}
\end{figure}

The following theorem, essentially proved in~\cite{Fellows03,ChorFJ04}, shows that given a sufficiently large graph with no isolated vertices, it is possible to efficiently find either a matching of a prescribed size or a crown decomposition (see also~\cite[Lemma~4.5]{KernelBook19}).
\begin{theorem}[\cite{Fellows03,ChorFJ04}]\label{thm:crown_algo}
There exists a polynomial-time algorithm, called $\Crown$, that given a graph $G$ and an integer $k \geq 1$, such that $G$ has at least $3k-2$ vertices and has no isolated vertices, finds either a matching of $G$ of size $k$ or a crown decomposition of $G$.
\end{theorem}

\section{Kernels for Storage Capacity and Dual Index Coding}\label{sec:kernels}

In this section, we present and analyze the kernelization algorithm for the $\SC_q$, $\DIC_q$, and $\DMR_\Fset$ problems and prove Theorems~\ref{thm:IntroKernelCap},~\ref{thm:IntroKernelInd}, and~\ref{thm:IntroKernelMR}. We also derive the fixed-parameter tractability of these problems and confirm Theorem~\ref{thm:IntroFPT}.

\subsection{Reduction Rules}

Our kernelization algorithm uses two reduction rules, iteratively applied to decrease the size of the graph at hand.
The following lemmas analyze the effect of these rules on the graph quantities considered in the present paper.
The first reduction rule concerns removing isolated vertices from the graph.

\begin{lemma}\label{lemma:isolated}
Let $G=(V,E)$ be a graph with an isolated vertex $u \in V$, and let $G' = G[V \setminus \{u\}]$.
Then
\begin{enumerate}
  \item\label{itm:CapI} for every integer $q \geq 2$, $\Capa_q(G) = \Capa_q(G')$,
  \item\label{itm:IndI} for every integer $q \geq 2$, $\Ind_q(G) = \Ind_q(G')+1$, and
  \item\label{itm:minrkI} for every field $\Fset$, ${\minrank}_\Fset(G) = {\minrank}_\Fset(G')+1$.
\end{enumerate}
\end{lemma}

\begin{proof}
For Items~\ref{itm:CapI} and~\ref{itm:IndI}, fix some integer $q \geq 2$, and consider the confusion graphs $\Conf_q(G)$ and $\Conf_q(G')$. Observe that $\Conf_q(G')$ is isomorphic to the subgraph of $\Conf_q(G)$ induced by the vertices that have some fixed value in the entry associated with the vertex $u$. Since $u$ is isolated in $G$, every two vertices of $\Conf_q(G)$ with distinct values in the entry associated with $u$ are adjacent. It therefore follows that $\Conf_q(G)$ consists of $q$ disjoint copies of $\Conf_q(G')$, with all the possible edges between vertices of distinct copies. This implies that $\alpha(\Conf_q(G)) = \alpha(\Conf_q(G'))$ and $\chi(\Conf_q(G)) = \chi(\Conf_q(G')) \cdot q$. By Propositions~\ref{prop:CapConf} and~\ref{prop:IndConf}, we derive that $\Capa_q(G) = \Capa_q(G')$ and $\Ind_q(G) = \Ind_q(G')+1$.

For Item~\ref{itm:minrkI}, fix a field $\Fset$, and notice that every matrix that represents the graph $G$ over $\Fset$ is obtained by combining a matrix that represents $G'$ over $\Fset$ with a row and a column associated with $u$, all of whose entries but the diagonal are zeros. This implies that ${\minrank}_\Fset(G) = {\minrank}_\Fset(G')+1$ and completes the proof.
\end{proof}

Before turning to the second reduction rule, let us state and prove the following lemma (see~\cite{Mazumdar15}).

\begin{lemma}\label{lemma:matching}
Let $G$ be a graph on $n$ vertices that admits a matching of size $r$. Then
\begin{enumerate}
  \item\label{itm:CapM} for every integer $q \geq 2$, $\Capa_q(G) \geq r$,
  \item\label{itm:IndM} for every integer $q \geq 2$, $\Ind_q(G) \leq n-r$, and
  \item\label{itm:minrkM} for every field $\Fset$, ${\minrank}_\Fset(G) \leq n-r$.
\end{enumerate}
\end{lemma}

\begin{proof}
Let $G$ be a graph on $n$ vertices with a matching $M$ of size $r=|M|$.
The graph $G$ admits a clique cover of size $n-r$, as follows by combining the $r$ edges of the matching $M$, viewed as cliques of size $2$, together with the remaining $n-2r$ vertices that do not lie in $M$, viewed as cliques of size $1$.
Using Lemmas~\ref{lemma:IndCliqueCover} and~\ref{lemma:DualityIndCap}, it follows that for every integer $q \geq 2$, it holds that
\[ n-{\Capa}_q(G) \leq {\Ind}_q(G) \leq \overline{\chi}(G) \leq n-r,\]
as desired for Items~\ref{itm:CapM} and~\ref{itm:IndM} of the lemma.
For Item~\ref{itm:minrkM}, we use Lemma~\ref{lemma:MRCliqueCover} to obtain that for every field $\Fset$, ${\minrank}_\Fset(G) \leq \overline{\chi}(G) \leq n-r$.
This completes the proof.
\end{proof}

The second and main reduction rule of our kernelization algorithm replaces a graph $G$ by its subgraph $G[R]$ induced by the set $R$ from a given crown decomposition of $G$ (see Definition~\ref{def:crown}).

\begin{lemma}\label{lemma:crown}
Let $G = (V,E)$ be a graph with a crown decomposition $V = C \cup H \cup R$. Then
\begin{enumerate}
  \item\label{itm:Cap} for every integer $q \geq 2$, $\Capa_q(G) = \Capa_q(G[R])+|H|$,
  \item\label{itm:Ind} for every integer $q \geq 2$, $\Ind_q(G) = \Ind_q(G[R])+|C|$, and
  \item\label{itm:minrk} for every field $\Fset$, ${\minrank}_\Fset(G) = {\minrank}_\Fset(G[R])+|C|$.
\end{enumerate}
\end{lemma}

\begin{proof}
Let $G = (V,E)$ be a graph on the vertex set $V=[n]$, and let $V = C \cup H \cup R$ be a crown decomposition of $G$.
By definition, the graph $G$ admits a matching of size $|H|$, connecting each vertex of $H$ to a vertex of $C$.

For Item~\ref{itm:Cap}, fix an integer $q \geq 2$, and let us prove that $\Capa_q(G) = \Capa_q(G[R])+|H|$.
We first prove that $\Capa_q(G) \geq \Capa_q(G[R])+|H|$.
The graph $G[C \cup H]$ admits a matching of size $|H|$, hence by Item~\ref{itm:CapM} of Lemma~\ref{lemma:matching}, it holds that ${\Capa}_q(G[C \cup H]) \geq |H|$. Using Lemma~\ref{lemma:union}, this implies that
\[ {\Capa}_q(G) \geq {\Capa}_q(G[R]) + {\Capa}_q(G[C \cup H]) \geq {\Capa}_q(G[R]) + |H|.\]

We next show that $\Capa_q(G) \leq \Capa_q(G[R])+|H|$.
By Proposition~\ref{prop:CapConf}, this inequality is equivalent to $\alpha(\Conf_q(G)) \leq \alpha(\Conf_q(G[R])) \cdot q^{|H|}$.
Let $I$ be a maximum independent set in $\Conf_q(G)$. By the pigeonhole principle, there exists a set $I' \subseteq I$ of size $|I'| \geq |I| \cdot q^{-|H|}$ such that the vectors of $I'$ share the same restriction to the entries of $H$.
We claim that all the vectors of $I'$ agree on the entries of $C$.
Indeed, fix some $i \in C$, and notice that $N_G(i) \subseteq H$, because $C$ is an independent set in $G$ with no edges to the vertices of $R$. Hence, for every pair of distinct vectors $x,y \in I'$, we have $x_{N_G(i)} = y_{N_G(i)}$, which implies that $x_i = y_i$, because $I'$ is an independent set in $\Conf_q(G)$.
Now, let $J$ denote the set of restrictions of the vectors of $I'$ to the entries of $R$. Since the vectors of $I'$ agree on the entries of $C \cup H$, it holds that $|I'| = |J|$. We will prove now that the set $J$ forms an independent set in $\Conf_q(G[R])$. This will complete the proof, because it implies that
\[\alpha(\Conf_q(G)) = |I| \leq |I'| \cdot q^{|H|} =  |J| \cdot q^{|H|} \leq \alpha(\Conf_q(G[R])) \cdot q^{|H|}.\]
To prove that $J$ forms an independent set in $\Conf_q(G[R])$, let $x,y \in J$ be two distinct vectors. By the definition of $J$, there exist two vectors $\widetilde{x}, \widetilde{y} \in I$ such that $\widetilde{x}_{C \cup H} = \widetilde{y}_{C \cup H}$, $\widetilde{x}_{R} = x$, and $\widetilde{y}_{R} = y$. Since $\widetilde{x}$ and $\widetilde{y}$ are not adjacent in $\Conf_q(G)$, it follows that there is no $i \in [n]$ with $\widetilde{x}_i \neq \widetilde{y}_i$ and $\widetilde{x}_{N_G(i)} = \widetilde{y}_{N_G(i)}$. By $\widetilde{x}_{C \cup H} = \widetilde{y}_{C \cup H}$, this implies that there is no $i \in R$ such that $\widetilde{x}_i \neq \widetilde{y}_i$ and $\widetilde{x}_{N_{G[R]}(i)} = \widetilde{y}_{N_{G[R]}(i)}$, hence $x$ and $y$ are not adjacent in $\Conf(G[R])$, as required.

For Item~\ref{itm:Ind}, fix an integer $q \geq 2$, and let us prove that $\Ind_q(G) = \Ind_q(G[R])+|C|$.
We first show that $\Ind_q(G) \geq \Ind_q(G[R])+|C|$.
It suffices to prove that $\chi(\Conf_q(G)) \geq \chi(\Conf_q(G[R])) \cdot q^{|C|}$, which by Proposition~\ref{prop:IndConf}, implies that
\[{\Ind}_q(G) = \lceil \log_q \chi(\Conf_q(G)) \rceil \geq \lceil \log_q \chi(\Conf_q(G[R])) \rceil +|C| = {\Ind}_q(G[R])+|C|.\]
To this end, let $\mathfrak{B}$ denote the subgraph of $\Conf_q(G)$ induced by the vectors whose restriction to the entries of $H$ is, say, a vector of ones.
Observe that the subgraph of $\mathfrak{B}$ induced by the vertices that share some fixed restriction to the entries of $C$ is isomorphic to $\Conf_q(G[R])$.
Observe further that if two vertices of $\mathfrak{B}$ differ on their restrictions to the entries of $C$, then they are adjacent, because $C$ is an independent set in $G$ with no edges to the vertices of $R$.
This implies that $\mathfrak{B}$ consists of $q^{|C|}$ disjoint copies of $\Conf_q(G[R])$, with all the possible edges between vertices of distinct copies.
It therefore follows that $\chi(\Conf_q(G)) \geq \chi(\mathfrak{B}) = \chi(\Conf_q(G[R])) \cdot q^{|C|}$, as required.

We next show that $\Ind_q(G) \leq \Ind_q(G[R])+|C|$.
The graph $G[C \cup H]$ admits a matching of size $|H|$, hence by Item~\ref{itm:IndM} of Lemma~\ref{lemma:matching}, it holds that ${\Ind}_q(G[C \cup H]) \leq |C \cup H| - |H| = |C|$. Using Lemma~\ref{lemma:union}, this implies that ${\Ind}_q(G) \leq {\Ind}_q(G[R]) + {\Ind}_q(G[C \cup H]) \leq {\Ind}_q(G[R]) + |C|$.

For Item~\ref{itm:minrk}, fix a field $\Fset$, and let us prove that ${\minrank}_\Fset(G) = {\minrank}_\Fset(G[R])+|C|$.
We first show that ${\minrank}_\Fset(G) \geq {\minrank}_\Fset(G[R])+|C|$.
Let $A$ be a matrix that represents the graph $G$ over $\Fset$ and satisfies ${\minrank}_\Fset(G) = {\rank}_\Fset(A)$.
Consider the sub-matrices $A_R$ and $A_{C \cup R}$ of $A$ obtained by restricting $A$ to the rows and columns associated with the vertices of $R$ and $C \cup R$ respectively.
The matrix $A_R$ represents the graph $G[R]$ over $\Fset$, hence ${\rank}_\Fset(A_R) \geq {\minrank}_\Fset(G[R])$, and the matrix $A_{C \cup R}$ clearly satisfies ${\rank}_\Fset(A) \geq {\rank}_\Fset(A_{C \cup R})$.
Since $C$ is an independent set in $G$ with no edges to the vertices of $R$, it follows that up to a permutation of the rows and columns, $A_{C \cup R}$ is a $2 \times 2$ block matrix, where the diagonal block associated with $R$ is $A_R$, the diagonal block associated with $C$ is some diagonal matrix with nonzero values on the diagonal, and all the other entries are zeros. This implies that ${\rank}_\Fset(A_{C \cup R}) = {\rank}_\Fset(A_{R}) + |C|$. We therefore obtain that
\[ {\minrank}_\Fset(G) = {\rank}_\Fset(A) \geq {\rank}_\Fset(A_{C \cup R}) = {\rank}_\Fset(A_{R}) + |C| \geq {\minrank}_\Fset(G[R])+|C|.\]

We finally show that ${\minrank}_\Fset(G) \leq {\minrank}_\Fset(G[R])+|C|$.
The graph $G[C \cup H]$ admits a matching of size $|H|$, hence by Item~\ref{itm:minrkM} of Lemma~\ref{lemma:matching}, we have ${\minrank}_\Fset(G[C \cup H]) \leq |C \cup H|-|H| =  |C|$. Using Lemma~\ref{lemma:union}, this implies that
\[{\minrank}_\Fset(G) \leq {\minrank}_\Fset(G[R]) + {\minrank}_\Fset(G[C \cup H]) \leq {\minrank}_\Fset(G[R]) +|C|,\]
completing the proof.
\end{proof}

\subsection{The Kernelization Algorithm}\label{sec:ker_algo}

The proofs of Theorems~\ref{thm:IntroKernelCap},~\ref{thm:IntroKernelInd}, and~\ref{thm:IntroKernelMR} are based on the following kernelization algorithm.
Note that the input $(G,k)$ is updated throughout the run of the algorithm.
Recall that $K_0$ stands for the graph with no vertices and no edges.

\begin{tcolorbox}
Input: A pair $(G,k)$ of a graph $G=(V,E)$ and an integer $k \geq 0$.
\begin{enumerate}
  \item\label{itm:k=0}{If $k \leq 0$, then return $(K_0,0)$.}
  \item\label{itm:isolated} Remove from $G$ its isolated vertices.
  \item\label{itm:crown} If $G$ has at least $3k-2$ vertices, call the $\Crown$ algorithm from Theorem~\ref{thm:crown_algo} on $(G,k)$.
    \begin{itemize}
      \item If $\Crown$ returns a matching of $G$ of size $k$, then return $(K_0,0)$.
      \item If $\Crown$ returns a crown decomposition $V = C \cup H \cup R$ of $G$, then proceed with the instance $(G[R], k-|H|)$, and go back to Step~\ref{itm:k=0}.
    \end{itemize}
  \item\label{itm:return} Return $(G,k)$.
\end{enumerate}
\end{tcolorbox}

We first analyze the algorithm with respect to the $\SC_q$ problem and prove Theorem~\ref{thm:IntroKernelCap}.

\begin{proof}[ of Theorem~\ref{thm:IntroKernelCap}]
Fix $q \geq 2$, and consider an instance $(G,k)$ of the $\SC_q$ problem, where $G=(V,E)$ is a graph and $k$ is an integer.
Recall that this problem asks to decide whether $\Capa_q(G) \geq k$.
We prove the correctness of the kernelization algorithm presented above.

In Step~\ref{itm:k=0} of the algorithm, if $k \leq 0$, the algorithm correctly returns the fixed YES instance $(K_0,0)$.
In Step~\ref{itm:isolated} of the algorithm, the isolated vertices are removed from $G$, and the value of $k$ is not changed.
By Item~\ref{itm:CapI} of Lemma~\ref{lemma:isolated}, the removal of an isolated vertex from $G$ does not change the value of $\Capa_q(G)$, hence the equivalence of the instances is preserved.

Once $G$ has no isolated vertices, if it has at least $3k-2$ vertices, one can call the $\Crown$ algorithm from Theorem~\ref{thm:crown_algo} on $(G,k)$, as is done in Step~\ref{itm:crown} of our algorithm.
The theorem guarantees that $\Crown$ returns either a matching of $G$ of size $k$ or a crown decomposition $V = C \cup H \cup R$ of $G$. In the former case, Item~\ref{itm:CapM} of Lemma~\ref{lemma:matching} implies that $\Capa_q(G) \geq k$, hence the algorithm correctly returns a fixed YES instance. In the latter case, by Item~\ref{itm:Cap} of Lemma~\ref{lemma:crown}, it holds that ${\Capa}_q(G) = {\Capa}_q(G[R])+|H|$, ensuring the equivalence of the $(G,k)$ instance to the $(G[R], k-|H|)$ instance with which our algorithm proceeds. The algorithm repeats this process as long as the number of vertices in $G$ is at least $3k-2$ for the current value of the parameter $k$.
Since in each iteration the number of vertices of $G$ decreases, the number of iterations is bounded by the number of vertices of the input graph.

Finally, Step~\ref{itm:return} of the algorithm returns an instance $(G',k')$, equivalent to the original one, such that $G'$ has at most $3k'-3$ vertices. Since the algorithm never increases the value of the parameter $k$, it holds that $k' \leq k$, as required.

We conclude with the observation that the running time of the algorithm is polynomial, because the number of iterations is bounded by the number of vertices of the input graph and because the $\Crown$ algorithm runs, by Theorem~\ref{thm:crown_algo}, in polynomial time.
\end{proof}

We turn to the analysis of the kernelization algorithm for the $\DIC_q$ problem and prove Theorem~\ref{thm:IntroKernelInd}.

\begin{proof}[ of Theorem~\ref{thm:IntroKernelInd}]
Fix $q \geq 2$, and consider an instance $(G,k)$ of the $\DIC_q$ problem, where $G=(V,E)$ is a graph and $k$ is an integer.
Recall that this problem asks to decide whether $\Ind_q(G) \leq |V|-k$.
We prove the correctness of the kernelization algorithm presented at the beginning of the section.

In Step~\ref{itm:k=0} of the algorithm, if $k \leq 0$, the algorithm correctly returns a fixed YES instance.
In Step~\ref{itm:isolated} of the algorithm, the isolated vertices are removed from $G$, and the value of $k$ is not changed.
Every removal of an isolated vertex from $G$ decreases the number of vertices by $1$, and by Item~\ref{itm:IndI} of Lemma~\ref{lemma:isolated}, it also decreases the value of $\Ind_q(G)$ by $1$, hence the equivalence of the instances is preserved.

Once $G$ has no isolated vertices, if it has at least $3k-2$ vertices, one can call the $\Crown$ algorithm from Theorem~\ref{thm:crown_algo} on $(G,k)$, as is done in Step~\ref{itm:crown} of our algorithm.
The theorem guarantees that $\Crown$ returns either a matching of $G$ of size $k$ or a crown decomposition $V = C \cup H \cup R$ of $G$. In the former case, Item~\ref{itm:IndM} of Lemma~\ref{lemma:matching} implies that $\Ind_q(G) \leq |V|-k$, hence the algorithm correctly returns a fixed YES instance. In the latter case, by Item~\ref{itm:Ind} of Lemma~\ref{lemma:crown}, it holds that $\Ind_q(G) = \Ind_q(G[R])+|C|$, ensuring that the inequality $\Ind_q(G) \leq |V|-k$ holds if and only if it holds that
\[{\Ind}_q(G[R]) = {\Ind}_q(G)-|C| \leq |V|-|C|-k = |R|-(k-|H|),\]
hence the $(G,k)$ instance is equivalent to the $(G[R], k-|H|)$ instance with which our algorithm proceeds. The algorithm repeats this process as long as the number of vertices in $G$ is at least $3k-2$ for the current value of the parameter $k$.
Since in each iteration the number of vertices of $G$ decreases, the number of iterations is bounded by the number of vertices of the input graph.

Finally, Step~\ref{itm:return} of the algorithm returns an instance $(G',k')$, equivalent to the original one, such that $G'$ has at most $3k'-3$ vertices. Since the algorithm never increases the value of the parameter $k$, it holds that $k' \leq k$, as required. As shown in the proof of Theorem~\ref{thm:IntroKernelCap}, the running time of the algorithm is polynomial, so we are done.
\end{proof}

The proof of Theorem~\ref{thm:IntroKernelMR} is essentially identical to the one of Theorem~\ref{thm:IntroKernelInd}. The only difference lies in applying the third items of Lemmas~\ref{lemma:isolated},~\ref{lemma:matching}, and~\ref{lemma:crown} instead of their second items. To avoid repetitions, we omit the details.

\subsection{Fixed-Parameter Algorithms}

Now, we combine our kernelization results with Proposition~\ref{prop:exp_algo} to obtain the fixed-parameter algorithms declared in Theorem~\ref{thm:IntroFPT}.
\begin{proof}[ of Theorem~\ref{thm:IntroFPT}]
Fix $q \geq 2$, and let $(G,k)$ be an instance of the $\SC_q$ problem, where $G$ is a graph on $n$ vertices and $k$ is an integer.
Consider the algorithm that first calls the kernelization algorithm given by Theorem~\ref{thm:IntroKernelCap} on $(G,k)$ to obtain a kernel $(G',k')$, where $G'$ has at most $3k$ vertices, and then calls the algorithm given by Item~\ref{alg:C} of Proposition~\ref{prop:exp_algo} on $(G',k')$. The correctness of the algorithm immediately follows from the correctness of the used algorithms. The running time of the kernelization step is polynomial in $n$, and the running time of the decision step is $\widetilde{O}(1.1996^{q^{3k}})$. The total running time is thus bounded by $n^{O(1)} \cdot \widetilde{O}(1.1996^{q^{3k}})$, completing the proof of Item~\ref{itm:FPT1} of the theorem.

Items~\ref{itm:FPT2} and~\ref{itm:FPT3} are proved similarly. For Item~\ref{itm:FPT2}, one has to combine Theorem~\ref{thm:IntroKernelInd} with Item~\ref{alg:I} of Proposition~\ref{prop:exp_algo}, and for Item~\ref{itm:FPT3}, one has to combine Theorem~\ref{thm:IntroKernelMR} with Item~\ref{alg:MR} of Proposition~\ref{prop:exp_algo}.
To avoid repetitions, we omit the details.
\end{proof}

\section*{Acknowledgments}
We thank the anonymous referees for their valuable suggestions.

\bibliographystyle{abbrv}
\bibliography{crown}

\end{document}